\documentclass[pra,superscriptaddress,twocolumn,longbibliography,reprint]{revtex4-2}
\usepackage{amsmath,amsfonts,amssymb,graphics,graphicx,epsfig,color,times,natbib,amsthm,bm,makecell,soul}
\usepackage{booktabs}
\usepackage{dcolumn}
\usepackage{dsfont}
\usepackage{subfigure}
\usepackage{hyperref}
\usepackage{appendix}
\usepackage[ruled]{algorithm2e}
\newtheorem{lemma}{Lemma}
\newtheorem{theorem}{Theorem}

\newcommand{\sign}[1]{\mathrm{sgn}(#1)}

\hypersetup{colorlinks=true, citecolor=blue, urlcolor=blue, linkcolor=blue}
\begin{document}
	\date{\today}
	\title{Quantum Classifiers with Trainable Kernel}
	\author{Li Xu}
	\affiliation{College of Science, China University of Petroleum, 266580 Qingdao, P.R. China.}
	\author{Xiao-yu Zhang}
	\affiliation{College of Science, China University of Petroleum, 266580 Qingdao, P.R. China.}
	\author{Ming Li}
	\email{liming@upc.edu.cn.}
	\affiliation{College of Science, China University of Petroleum, 266580 Qingdao, P.R. China.}
    \author{Shu-qian Shen}
	\affiliation{College of Science, China University of Petroleum, 266580 Qingdao, P.R. China.}
	\begin{abstract}
	 Kernel function plays a crucial role in machine learning algorithms such as classifiers.  In this paper, we aim to improve the classification performance and reduce the reading out burden of quantum classifiers. We devise a universally trainable quantum feature mapping layout to broaden the scope of feature states and avoid the inefficiently straight preparation of quantum superposition states. We also propose an improved quantum support vector machine that employs partially evenly weighted trial states. In addition, we analyze its error sources and superiority. As a promotion,  we propose a  quantum iterative multi-classifier framework for one-versus-one and one-versus-rest approaches. Finally, we conduct corresponding numerical demonstrations in the \textit{qiskit} package. The simulation result of trainable quantum feature mapping shows considerable clustering performance, and the subsequent classification performance is superior to the existing quantum classifiers in terms of accuracy and distinguishability.
	\end{abstract}
	\maketitle

	\section{Introduction}
    Quantum algorithms have been widely studied in the past few decades, with the hope of overcoming some problems that are intractable to classical computers. Notable examples include Shor's  quantum factoring algorithm \cite{Shor}, Grover's search  algorithm \cite{grover}, and the quantum algorithm for linear systems of equations proposed by Harrow, Hassidim, and Lloyd(HHL algorithm) \cite{HHL}. At the same time, machine learning has exploded into a hot research topic. The intersection between machine learning and quantum computation has burst an interdisciplinary field, quantum machine learning, attracting massive academic research and application during the past decade \cite{QML,QML1,QPCA,QML2}.

     Support vector machine(SVM) is a representative supervised machine learning algorithm aiming to return an optimal hyperplane that separates two classes of samples with distinct labels and then assigns a label for a new datum determined by which side it lies {\cite{SVM,svm}}. Drawing support from the parallelism of quantum algorithms, quantum machine learning algorithms could surpass their classical counterparts, especially in processing big data. Specifically, the least square quantum support vector machine(LS-QSVM) inherits a logarithmic level complexity by means of superpositions and  HHL algorithm \cite{QSVM}.  Meanwhile, the experiment verifies a proof-of-principle task \cite{EQSVM}. However, as the dataset size multiplicatively increases, the distinguishability of LS-QSVM becomes poor, which goes against `big data'(refer to Theorem \ref{th1}). To overcome this challenge, we propose the support vectors based quantum support vector machine(SV-QSVM). More importantly, if the data are fed classically, it will take $O(N)$ operations to load classical data into superposition, where $N$ represents the data dimension \cite{QRAM,qram,qram1}. This scaling can  predominate the complexity of a quantum algorithm and, thereby, impair potential quantum advantages.
	
	Recently, Schuld \textit{et al} and Havl\'{i}$\rm\check{c}$ek \textit{et al} have proposed an artful path that could circumvent the direct preparation of coherent superpositions by mapping the data to an exponential Hilbert feature space \cite{qfm,qfm1}. This kernel-based method stimulates the widespread applications of noisy intermediate-scale quantum(NISQ) computers \cite{nisq,nisq1,nisq2,nisq3,nisq4}. Here we collectively refer to the approaches only invoke classical optimizers in \cite{qfm,qfm1} as \textit{explicit approach}, and the approaches that utilize classical machine learning algorithms are referred to \textit{implicit approach}.  We briefly retrospect these two approaches and analyze their discrepancies and drawbacks. Both approaches demand the use of quantum feature mapping(QFM) to load classical data into quantum feature states. The explicit approach applies a parameterized quantum circuit(PQC) $W(\theta)$ optimized during the training process to the feature state, then measures the final states in the Z-basis. The classification result is revealed by some Boolean function $f:\{0,1\}^n\rightarrow\{1,-1\}$ through the output bit string, and the optimal parameters are acquired by minimizing the empirical risk function. However, the efficiency of the optimization phase is unsatisfactory, and more importantly, it cannot achieve a satisfactory classification success rate. In contrast, the \textit{implicit approach} utilizes quantum devices only for evaluating the kernel matrices and leaves the other assignments to classical computers. Hence, the \textit{implicit approach} inherits the accurate calculation from the classical computers and weakness---intractable with handling big data.

    To summarize, both approaches utilize QFM for loading classical information into quantum feature states, and then processing them in the Hilbert space. With an appropriate QFM, data can be repositioned ideally, i.e. data of the same category can be clustered, while data of different categories can be separated from each other.	When given an ill-conditioned QFM, for \textit{explicit approach}, it will increase the burden of PQC to reposition the locations of data linearly. Even sometimes PQC cannot fully complete the tasks due to linearity. In contrast, \textit{implicit approach} will directly return an ill-conditioned kernel matrix, since it thoroughly relies on QFM to construct a kernel matrix. As a consequence, the follow-up classical machine learning algorithm will show poor performance. Accordingly, finding a suitable QFM will bring significant benefits to both approaches. More generally, the properties of kernel matrices are crucial for machine learning algorithms that rely on kernel functions.

    Inspired by data re-uploading strategy \cite{re-up}, we combine the \textit{explicit approach} with \textit{implicit approach} to propose a trainable quantum feature mapping(TQFM) layout that owns multiple subsequent usability. On the one hand, a well-trained QFM circuit serves as the \textit{explicit approach}, and it owns a better data fitting effect because of nonlinear data re-uploading. Although the re-uploading model can be mapped to a linear explicit model, the overhead is high---$O(D\log D)$ additional qubits and gates, with $D$ the number of encoding gates used by the data re-uploading model \cite{re-up1}. On the other hand, the quantum feature states can be used to prepare kernel matrices, and then applied as subroutines to kernel based machine learning algorithms. Compared to the \textit{explicit approach}, which only relies on training parameters and no longer uses samples for classification. The latter requires both the parameters and samples, that is, to fully utilize the information of samples. Therefore, this method is particularly important for small sample learning.

	The structure of the paper is as follows: Section \ref{section2} presents the construction of trainable quantum feature mapping. The three subsequent use of TQFM is discussed simply in section \ref{section3}. The algorithm flow of SV-QSVM and error analysis is shown in Section \ref{section4}. In Section \ref{section5}, we discuss the essence of multi-classification algorithms and propose quantum iterative multi-classifiers.  The numerical simulations of Sections \ref{section2}, \ref{section3}, \ref{section4}, and \ref{section5} are exhibited in the four subsections of Section \ref{section6}, respectively. Section \ref{section7} concludes the primary work and future outlook.
	
	\section{Trainable quantum feature mapping}\label{section2}
	The quantum feature mappings are superior to classical counterparts in respect of expressiveness and get benefit from the exponential Hilbert space, hence machine learning algorithms with quantum kernels have the potential to achieve better performance than those with classical kernels. In this section, we propose the layout and loss function of trainable quantum feature mapping, and explore the optimal way to optimize the loss function.
    \subsection{Layout and loss function}
    The core idea of quantum kernel tricks that leads to its performance surpassing classical is the quantum kernel estimation, which maps $n$-dimensional classical data $x\in{\mathbb{R}}^n$ non-linearly to a quantum state $|\psi(x)\rangle=U(x)|0...0\rangle$ with the general dimension $2^n$. Here $U(x)$ is a PQC usually compiled by the feature of $x$ that decides the unique quantum feature mapping, and its layout affects the behavior of QFM.

    The central insight of quantum kernel is utilizing quantum circuits estimating the kernel functions $k^p(x_i,x_j)=|\langle{\psi(x_i)}|\psi(x_j)\rangle|^p$ in the Hilbert space, where $p$ is a positive integer. Why it is not possible to directly apply the PQC $W(\theta)$ of \textit{explicit approach} after $|\psi(x)\rangle$ and then estimate the kernel function is that $\langle{\psi(x_i)}|W^{\dagger}(\theta)W(\theta)|\psi(x_j)\rangle=\langle{\psi(x_i)}|\psi(x_j)\rangle$. Even if $W(\theta)$ possesses the capacity to reposition data,  it seems to do some useless work when calculating kernel functions. A common option to circumvent this situation is to re-upload data and integrate the layout of QFM with PQC, giving the formula of TQFM
    \begin{align}
    |\psi(x,\theta)\rangle&=U(x,\theta)|0\rangle^n,\notag\\
    U(x,\theta)&=U_l(x,\theta_l)U_{ent}...U_2(x,\theta_2)U_{ent}U_1(x,\theta_1),
    \end{align}
where $\theta$ is a family of parameters, and the functional relationship between $x$ and $\theta$ is user-specified, just like the activation function in neural networks. It should be pointed out that the construction of this function has different requirements depending on its subsequent use. If TQFM is directly used for classification (Section \ref{section3a}), more sophisticated functions are required, such as in \cite{re-up}. If TQFM is used for preparing ensembles or kernel matrices, a linear function is sufficient. The layout of $U(x,\theta)$ is shown in FIG.\ref{fig1a}.
    \subsection{Training stage}
\begin{figure}[t]
		\centering
        \subfigure[]{\includegraphics[width=0.4\textwidth]{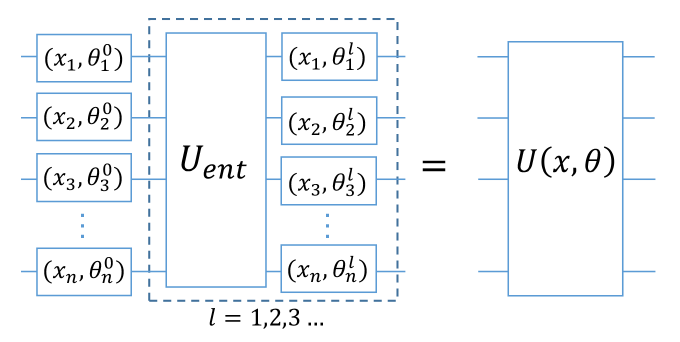}\label{fig1a}}
		\subfigure[]{\includegraphics[width=0.25\textwidth]{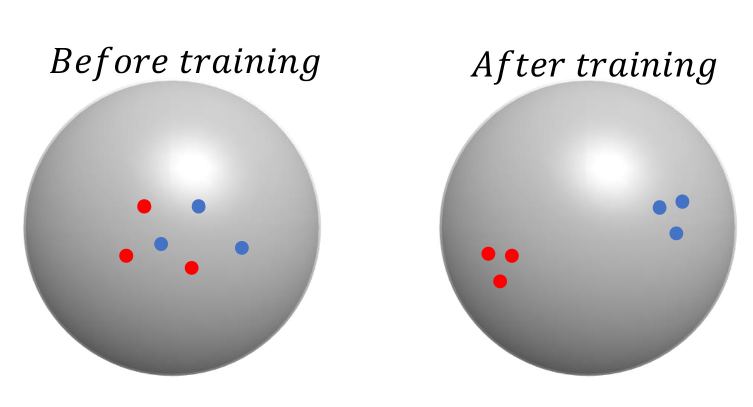}\label{fig1b}}\subfigure[]{\includegraphics[width=0.25\textwidth]{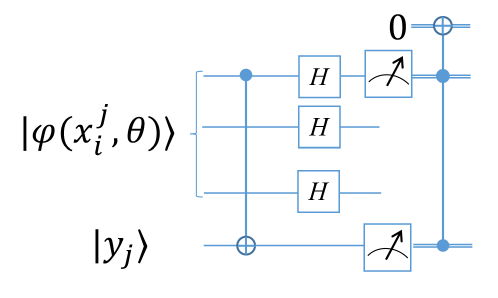}\label{fig1c}}
		\caption{(a) Circuits representation of $U(x,\theta)$. The single qubit gates $(x,\theta)$ are Pauli rotations. $U_{ent}$ are entangled quantum gates.  (b) Illustration of trainable quantum feature mapping. After training, samples of the same category are clustered together, while samples of different categories are orthogonal to each other. (c) Optimal circuit diagram for optimizing loss function. The inner product is given by the probability difference of the classical bit. }
	\end{figure}

    After introducing $\theta$, it is the turn to optimize these parameters to cluster homogeneous data to improve the performance of quantum kernel functions(FIG.\ref{fig1b}). Let $|y_j\rangle_{j=1}^L$ be a real label vector, and $\langle y_i|y_j\rangle=\delta_{ij}$, where $\delta_{ij}$ is the Kronecker delta function. Assuming that each $y_j$ corresponds to $M_j$ training samples. We minimize the following loss function
    \begin{align}
    E(\theta) = 1- \frac{1}{L}\sum_{j=1}^L\frac{1}{M_j}\sum_{i=1}^{M_j}|\langle\psi(x_i^j,\theta)|y_j\rangle|^2\label{e2},
    \end{align}
where $x_i^j$ is the $i$-th data in class $j$.

    Take binary classification as an example, any binary classifier can be transformed into a multi-class
classifier using the one-versus-one or one-versus-rest strategy. Minimizing Eq.(\ref{e2}) is equivalent to maximizing the value of $|\langle\psi(x_i,\theta)|\psi(x_j,\theta)\rangle|^2$ when $x_i$ and $x_j$ belong to the same category and minimizing the value when not. To calculate such an inner product using a quantum computer, one simply applies $U(x_i,\theta)$ and $U(x_j,\theta)$ to the initial state $|0\rangle^n$ individually and then perform a SWAP test. Each SWAP operator needs $2n+1$ qubits, $l$-layer PQC, and $1$ set of measuring devices. As shown in Ref.\cite{qfm1}, if the embedded circuit $U(x,\theta)$ can be inverted, one implements a circuit $U^{\dagger}(x_i,\theta)U(x_j,\theta)|0\rangle^n$, and measures the overlap to the $|0\rangle^n$ state. This inversion test only needs $n$ qubits, $2l$-layer PQC, and $n$ sets of measuring devices overall.

    The authors of Ref.\cite{swap} have proposed multiple ways of optimizing the conventional SWAP test. After employing the novel protocol(FIG.\ref{fig1c}), the overhead of minimizing Eq.(\ref{e2}) can be decreased to $n+1$ qubits, $l$-layer PQC, and $2$ sets of measuring devices in total. More importantly, the reduced local loss function can be trained more effectively without suffering from barren plateaus. Supposing the labels of two categories are $1$ and $-1$ we can rewrite the label vector $|y_a\rangle=|\frac{1+a}{2}\rangle|\cdot\rangle^{n-1}$, with $a=\pm 1$. Removing the last $n-1$ qubits of label vector does not affect the orthogonality, it only has a certain impact on the compactness of clustering. Therefore, to save resources, swapping only the first qubit of $|\psi(x^a,\theta)\rangle$ with $|\frac{1+a}{2}\rangle$ is sufficient.

    \section{Usability of Trainable Quantum Feature Mapping}\label{section3}
    Compared to QFM, the trainability (the number of trainable parameters) and nonlinearity of TQFM give it more powerful modes and more optional subsequent uses.  Generally speaking, a TQFM with more parameters means stronger trainability. In this section, we present three subsequent options of TQFM in obtaining the optimal parameter $\theta^*$.
    \subsection{Explicit approach}\label{section3a}
    \textit{Explicit approach}, also known as \textit{quantum neural network} (QNN), consists of PQCs and gradient-based optimizers. TQFM can serve as a QNN since one just needs to compile new datum $x$ and $\theta^*$ into the TQFM framework, and then reveals the result by measuring the frequency of the label qubit.
    This is just one of the follow-up uses of TQFM, and we will not delve deeper into it here, as quantum neural networks have been widely studied  \cite{QNN1,QNN,QNN2,QNN3,QNN4}.
    \subsection{Ensemble model}
 As the training data has already been clustered, one can prepare the ensembles for the two categories \cite{qmc,seth}
    \begin{align*}
    \rho^a=\sum_{i=1}^{M_a}p_i^a\rho_i^a=\sum_{i=1}^{M_a}p_i^a|\psi(x_i^a,\theta^*)\rangle\langle\psi(x_i^a,\theta^*)|,
    \end{align*}
where $p_i^a\geqslant 0$ and $\sum_ip_i^a=1$, $a=\pm 1$. If the $l_1$ distance is chosen, the label for trial datum is given by
    \begin{align*}
    y &= \sign{\rm{Tr}(\rho^1\rho)-Tr(\rho^{-1}\rho)},
    \end{align*}
where $\rho=|\psi(x,\theta^*)\rangle\langle\psi(x,\theta^*)|$ embedded trial datum.
	 \subsection{Quantum kernel matrix}
 After undergoing TQFM stage, one can prepare a well-conditioned kernel matrix and then train an SVM for classification.  For brevity, the inner product of  $|\psi(x_i,\theta^*)\rangle$  and   $|\psi(x,\theta^*)\rangle$ is represented by $k(x_i,x)$, and $M=M_1+M_{-1}$. The classification expression of a general quantum SVM is as follows
    \begin{align}
    y=\sign{f(x)}=\sign{\sum_{i=1}^M\frac{\alpha_iy_ik(x_i,x)}{\sqrt{M}}},\label{e3}
    \end{align}
where $\alpha_i$ are hyperparameters of the classification hyperplane and meet the requirements  $\alpha_i\geqslant0$ and $\sum_i\alpha_i^2=1$. When facing diverse big data, Eq.(\ref{e3}) faces the dilemma of poor distinguishability.

    \begin{theorem}
Assuming that the number of samples M is a large number and trial datum $x$ is randomly chosen, then $f(x)=\sum_{i=1}^M\frac{\alpha_iy_ik(x_i,x)}{\sqrt{M}}$ is scaling to $O(\frac{1}{\sqrt M})$.\label{th1}
    \end{theorem}
    \begin{proof}
 Due to the randomness of $x$ and $M$ being a large number, we can deduce that $k(x_i,x)$ is independent and identically distributed with respect to $x_i$. Assuming $\alpha_i\propto\frac{1}{\sqrt M}$, we obtain that $\frac{\alpha_iy_ik(x_i,x)}{\sqrt{M}}$ is a distribution on the interval $[-\frac{1}{M},\frac{1}{M}]$. From the \textit{Central Limit Theorem}, it can be concluded that $f(x)$ roughly follows a normal distribution $N(0,\frac{1}{M})$.
    \end{proof}

      Performing a SWAP for $k$ times allows one to estimate the value to an accuracy $\pm\sqrt{|f(x)|(1-|f(x)|)k}$ \cite{seth}. As the absolute value of $f(x)$ decreases with $M$ increasing, maintaining the same accuracy of $f(x)$ requires more measurements. Unfortunately, the sign function is sensitive to values near 0. Ref.\cite{QSVM,nisq1} always encounter this challenge when processing big data, we will provide a solution to tackle this problem in the next section.

	\section{quantum support vector machine}\label{section4}
In the following, we analyze the advantages and disadvantages of SVM and LS-QSVM, and provide the detailed process of SV-QSVM algorithm.

    SVM is an efficient binary classification algorithm with high robustness, which can process high-dimensional data with the help of kernel functions, although it requires a complex training process for seeking out support vectors. In return, SVM has extremely high classification efficiency because it only relies on support vectors to classify new data. Equipped with a well-trained quantum kernel function, the classification error rate of SVM will be significantly reduced. Compared with quantum support vector machines, classical SVM has always been at a disadvantage in computing big data due to the lack of quantum parallelism. LS-QSVM overcomes the problem of computational scale, causing a decrease in distinguishability as the multiplication of samples due to treating all of them as support vectors.

We pick a compromise way---support vectors based variational quantum support vector machine---which can alleviate the matters of the lower computational scale of classical SVM and the poor distinguishability of quantum classifiers. The fundamental reason why SV-QSVM owns better distinguishability is that, like SVM, it only uses support vectors for classification, transforming the uniform superposition trial state related to sample size into support vectors dependent ones. This is equivalent to amplifying the amplitude of the trial state, that is, increasing the value of $f(x)$.

    \subsection{Support vector machine}
    To get a soft margin SVM in the Hilbert space $\mathcal{H}$, we consider the convex quadratic optimization program
	\begin{align}
		\min_{w,b,\xi}\ &\frac{1}{2}\|w\|^2+\frac{\gamma}{2}\sum_{i=1}^{M}\xi_i^2\label{e1}\\
		\rm{s.t.}\quad&y_i(\langle{w},\psi(x_i)\rangle_{\mathcal{H}}+b)\geqslant1-\xi_i,\notag
	\end{align}
	where $\xi_i\geqslant0$ are slack variables, $\gamma$ is a constant, $\{(x_i,y_i)\}_{i=1}^M\subseteq{\mathbb{R}}^n\times\{\pm1\}$ are training samples, and $w$ is a hyperplane in $\mathcal{H}$ that classifies data by decision function $y=\sign{\langle w,\psi(x)\rangle_{\mathcal{H}}+b}$. Note that the optimization program can be solved efficiently with a countable dimension of $\mathcal{H}$. However, once mapping to an exponential dimensional space, it becomes intractable to find the optimal hyperplane. The key issue that leads to the success of SVM and bypasses this obstacle is optimizing the dual program of Eq.(\ref{e1})
	\begin{align}
		\max_{\alpha}\ &\sum_{i=1}^M\alpha_i-\frac{1}{2}\sum_{i,j=1}^M\alpha_i\alpha_jy_iy_jk(x_i,x_j)-\frac{1}{2\gamma}\sum_{i=1}^M\alpha_i^2\label{e4}\\
		&{\rm{s.t.}}\quad\alpha_i\geqslant0,\quad\sum_{i=1}^M\alpha_iy_i=0, \notag
	\end{align}
	where $k(x_i,x_j)$ is the kernel function. As long as the kernel function can be effectively estimated, we can obtain a reasonable classifier by optimizing its dual form, giving an equivalent decision function with different expressions. Namely, $y=\sign{\sum_{i=1}^M\alpha_iy_ik\langle x_i,x\rangle+b}$, the bias $b$ can be replaced by $\sum_{i=1}^M\frac{\alpha_iy_i}{\lambda}$, and see Appendix \ref{apa} for the detailed derivation process.
    \subsection{Training of support vectors based quantum support vector machine}
    To make it applicable to the SV-QSVM model, we redraft the objective optimization function Eq.(\ref{e4}) in the following form
	\begin{align}
		L(\alpha)=\min_{\alpha}\ \frac{1}{2}\langle\alpha|K|\alpha\rangle-\||\alpha\rangle\|_1+Cl^2(\alpha),\label{e5}
	\end{align}
	where $K_{ij}$ is the entry of kernel matrix $K$ and $K_{ij}=y_iy_jk(x_i,x_j)+\delta_{ij}/{\gamma}$, $\||\alpha\rangle\|_1=\alpha_1+\alpha_2+...+\alpha_M$, $C$ is a penalty multiplier and $l(\alpha)=\sum_{i=1}^{M/2}\alpha_i-\sum_{M/2+1}^M\alpha_i$. The first constraint $\alpha_i\geqslant0$ can be automatically satisfied by limiting the angle of parameters, and the second constraint $\sum_{i=1}^M\alpha_iy_i=0$ is given as a squared penalty term.
	
	The first term of Eq.(\ref{e5}) is a variational quantum eigensolver(VQE) problem, a leading application for near-term devices \cite{vqe,vqe1,vqe2,vqe3,vqe4,vqe5}. VQE evaluates the ground state energy for a Hamiltonian which can be decomposed into a linear combination of unitary matrices by optimizing a PQC for generating the energy state. The most convenient representation basis is given by the Pauli operators. This is always efficient if the terms grow polynomially with qubit numbers. However, $K$ may not admit a Pauli decomposition within polynomial terms, i.e. the summation of terms grows exponentially with the system size. In this case, the random sampling trick (Appendix \ref{apb}) takes over the mission since persisting in using Pauli decomposition can cause inefficiency and therefore loss of quantum advantage \cite{md,md1}.

    The other two terms can be calculated by a Hadamard test circuit(FIG.\ref{fig2}). The probability of measurement on the ancilla qubit is
	\begin{align*}
		P_a=\frac{1}{2}[1+(-1)^a\langle0^{\otimes{m}}|H^{\otimes{m}}|U_{A}|0^{\otimes{m}}\rangle],
	\end{align*}
	with $a=0$ or $1$. The representation of $\||\alpha\rangle\|_1$ is given by $(P_0-P_1)$ multiplied by the dimensional factor $\sqrt{2^m}$. The formulation of $l(\alpha)$ can be acquired in the same way by adding the \textit{controlled-Z} gate.	
	
	\begin{figure}[t]
		\centering
		{\includegraphics[width=0.4\textwidth]{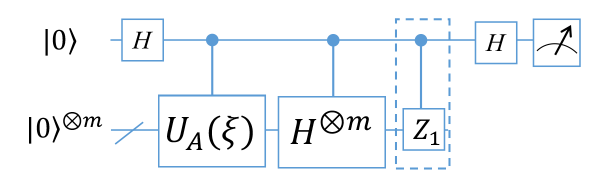}}
		\caption{Hadamard test for calculating $\||\alpha\rangle\|_1$(without C-Z gate) and $l(\alpha)$(with C-Z gate).  $U_A$ is PQC that prepares $|\alpha\rangle$, i.e. $|\alpha\rangle=U_{A}(\xi)|0\rangle^{\otimes{m}}$, with $m=\log M$. The subscript of \textit{Controlled-Z} gate represents the qubit applied.}\label{fig2}
	\end{figure}
    \subsection{Classification of support vectors based quantum support vector machine}
	The classification stage can be applied when finishing the training stage since the feedback of optimal $\xi^*$ is an essential prerequisite for classifying. Before building a classification circuit, we need to measure $U_{A}(\xi^*)$ to find out the positions of the support vectors, aiming at improving the distinguishability of trial data, as well as cutting down the circuit depth. The criterion for searching support vectors is whether the measurement frequency of the corresponding position is greater than a certain threshold, and non-zero values below this threshold are considered errors caused by noise.

	The classification result will be revealed by
	\begin{align*}
		y =\sign{f(x)}&=\sign{\langle\upsilon|\mu\rangle}=\sign{\sum_{i=1}^{M}\frac{\alpha_iy_ik(x_i,x)}{\sqrt{m_s}}}\\
|\mu\rangle&=\sum_{i=1}^M\alpha_iy_ic_i|i-1\rangle|\psi(x_i),\theta^*\rangle\\
|\upsilon\rangle&=\sum_{i=1}^M\frac{c_i|i-1\rangle|\psi(x,\theta^*)\rangle}{\sqrt{m_s}}
	\end{align*}
	where $m_s$ is the number of support vectors, and $c_i=1$ if and only if $|i-1\rangle$ corresponds to the positions of the support vectors, otherwise $c_i=0$. Under the premise that the locations of support vectors are found, one could improve the distinguishability relatively by substituting the evenly weighted  superposition $\sum_{i=1}^M\frac{|i-1\rangle|\psi(x,\theta^*)\rangle}{\sqrt M}$ with a partial evenly weighted state $|\upsilon\rangle$, since support vectors only account for a small portion. A commonly used method for this issue is  amplitude amplification, which applies the Grover operator $G=H^{\otimes m}OH^{\otimes m}T$ to the evenly weighted superposition $R^G\leqslant\lceil\frac{\pi}{4}\sqrt{\frac{M}{m_s}}\rceil$ times and then returns a partially weighted state \cite{grover,grover1,grover2}, where $T={\rm diag}((-1)^{c_1},(-1)^{c_2},...,(-1)^{c_M})$ and $O={\rm diag}(2,0,0,...,0)-I_M$. Using partial superposition $|\upsilon\rangle$ for classification can improve the distinguishability of classification results, which can be simply understood as $\frac{1}{\sqrt m_s}>\frac{1}{\sqrt M}$.

    The SV-QSVM gains rewards by constructing partial evenly weighted states that LS-QSVM cannot achieve, as it applies displacement to all samples, regarding them as support vectors. Although the solution process has been accelerated, at the expense of losing the choice to utilize partial evenly weighted states, making it difficult to extract the symbol of $\langle\upsilon|\mu\rangle$.   There are two commonly used methods to calculate whether the inner product is positive or negative. One is the Hadamard test which needs to construct controlled-unitary. The other one draws support from swap test, but we need an ancilla to construct $|\tilde{a}\rangle=1/\sqrt2(|0\rangle|0...0\rangle+|1\rangle|a\rangle)$ ($a=\upsilon$ or $\mu$) before testing.

	We reintroduce the fault tolerance classification equation
	\begin{align*}
		\sign{\langle\upsilon|\mu\rangle\mp\frac{1}{\sqrt k}}=\pm1.
	\end{align*}
to improve the robustness of the result and reduce the shot frequency \cite{mine}, where $k$ is the shot number.
		
	\subsection{Error analysis}
		\begin{figure*}[t]
		\centering
		\subfigure[]{\includegraphics[width=0.33\textwidth]{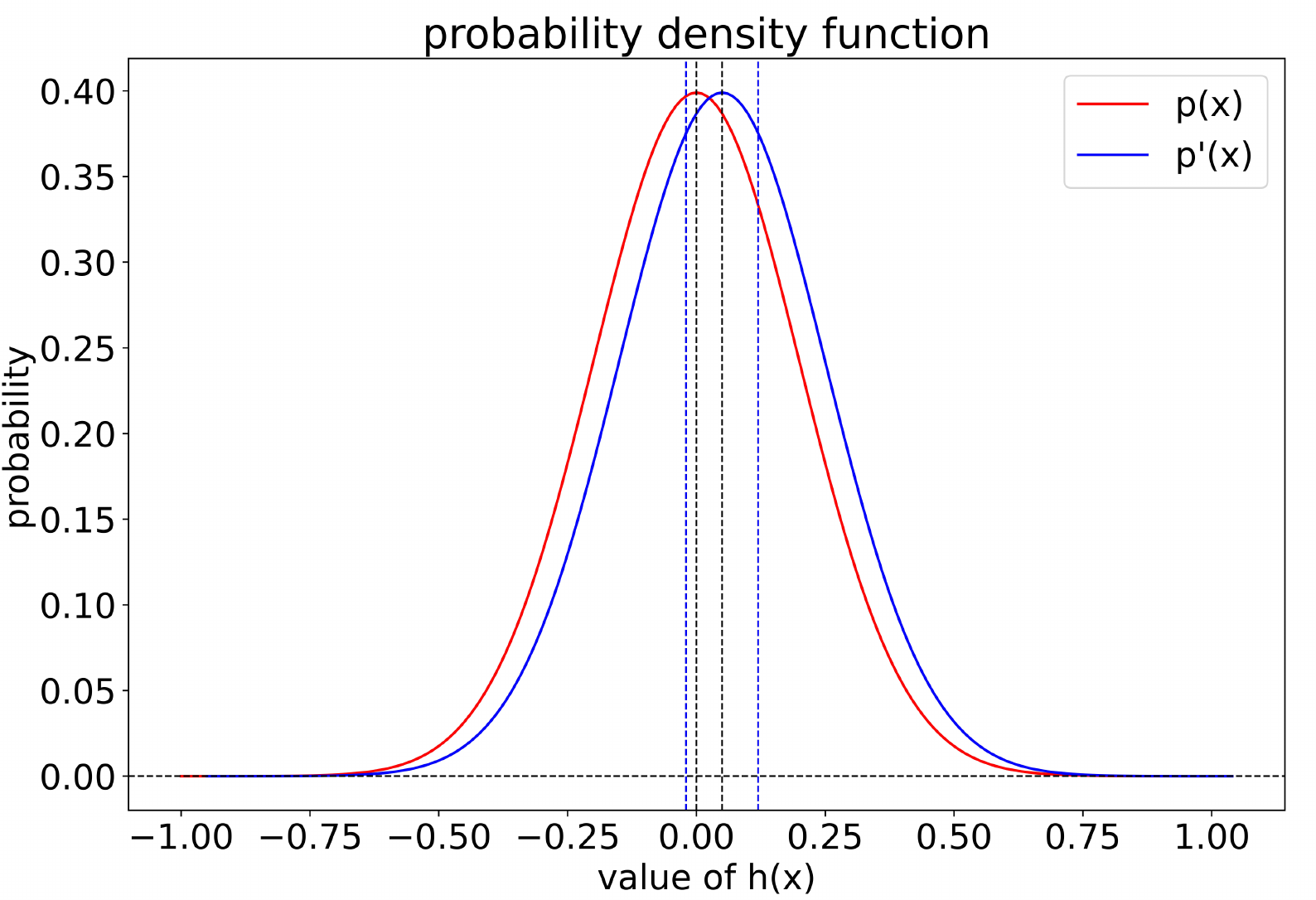}\label{fig3a}}
		\subfigure[]{\includegraphics[width=0.33\textwidth]{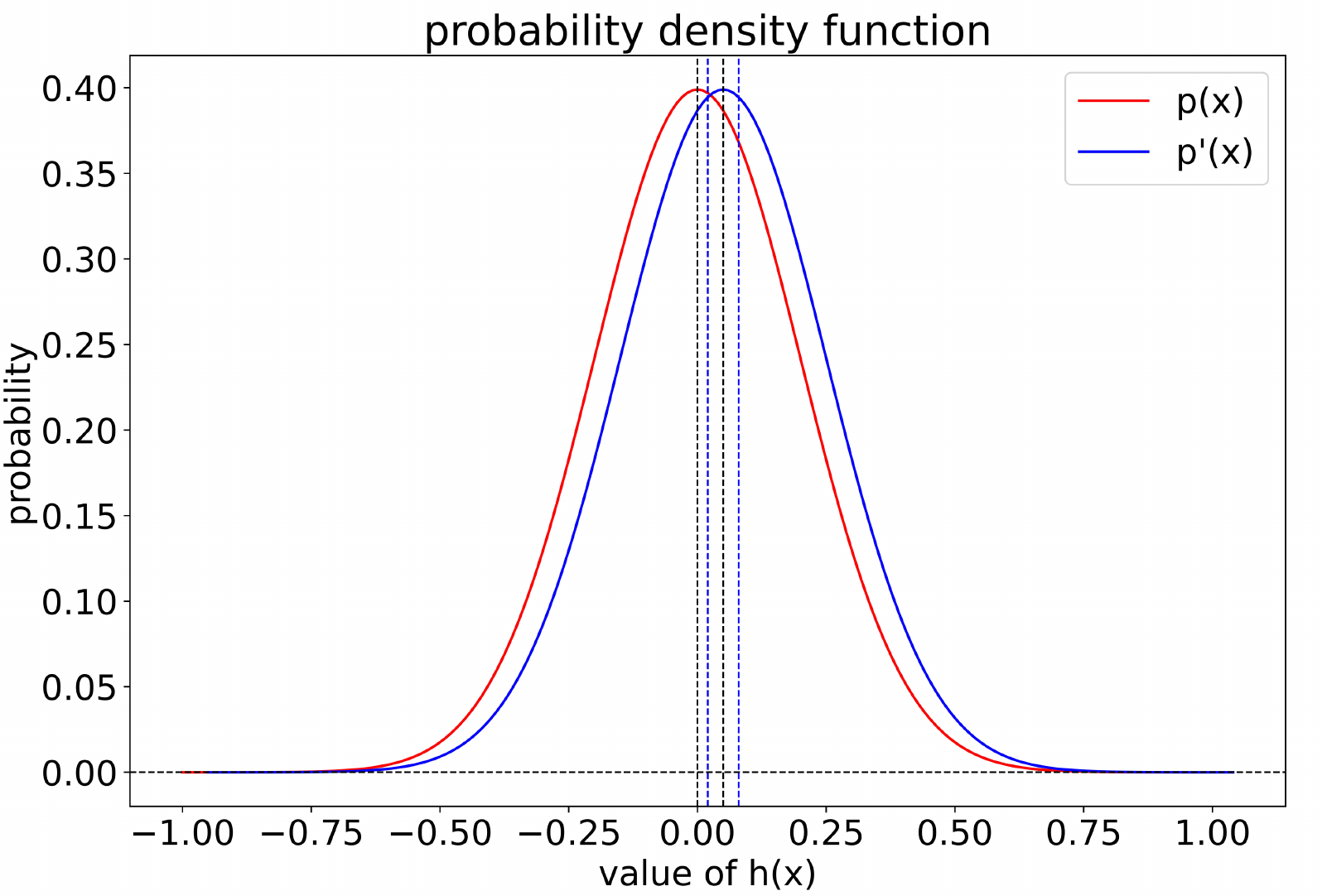}\label{fig3b}}
		\subfigure[]{\includegraphics[width=0.33\textwidth]{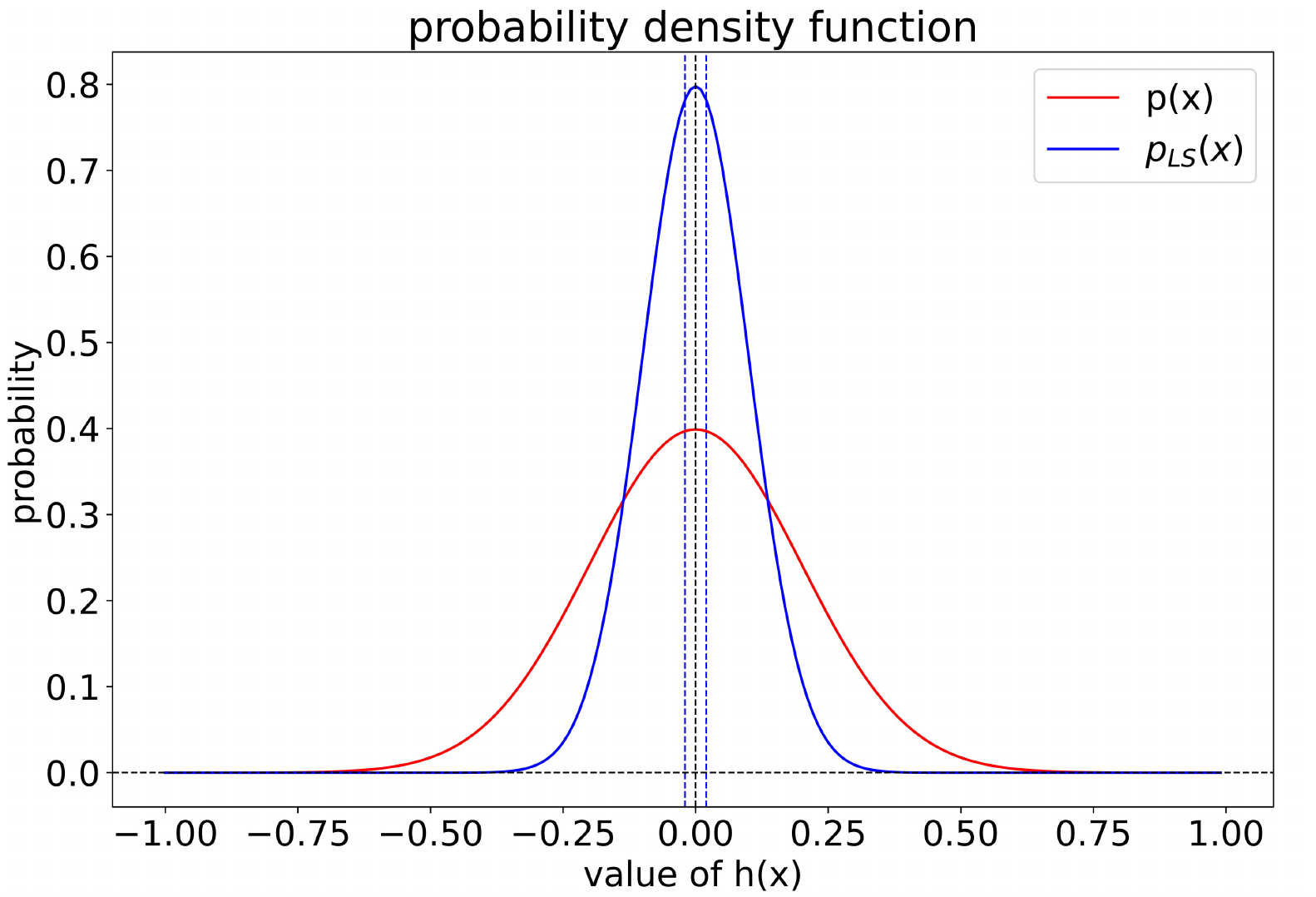}\label{fig3c}}
		\caption{Probability density functions in several different cases. The black dashed line represents the error between a noisy classifier and a noiseless classifier, while the blue dashed line represents the interval of the estimated value of the output of a noisy classifier. (a) Comparison between $p(x)$ and $p'(x)$ with $\varepsilon\geqslant\epsilon$. (b) Comparison between $p(x)$ and $p'(x)$ with $\varepsilon<\epsilon$. (c) Comparison between $p(x)$ and $p_{LS}(x)$. }
	\end{figure*}
	In this subsection, we rigorously analyze the error source of SV-QSVM and the advantages compared with LS-QSVM in \cite{QSVM}. For a brief review of LS-QSVM, see Appendix \ref{apc}. 
	\begin{lemma}\label{lemma1}{(Theorem 2.1 in \cite{lemma1})}  Let $x_0$ be the solution to the quadratic program
		\begin{align}
			\min\  &{\frac{1}{2}x^TKx-c^Tx}\label{5} \\
			\notag \rm{s.t.}\ &Gx\geqslant g,\\
			\notag &Dx = d,
		\end{align}
		where $K$ is noiseless and positive definite with the smallest eigenvalue $\lambda_{min}>0$. Let $K'$ be a noisy positive definite matrix such that  $\|K'-K\|_F\leqslant\varepsilon'<\lambda_{min}$. Let $x'_0$ be the solution to \rm{(\ref{5})} with $K$ replaced by $K'$. Then
		\begin{align*}
			\|x'_0-x_0\|_2\leqslant\frac{\varepsilon'}{\lambda_{min}-\varepsilon'}(1+\|x_0\|_2).
		\end{align*}
	\end{lemma}
	
	\begin{lemma}
		Let $h(x)=\sum_{i=1}^{m_s}\frac{\alpha_iy_ik(x,x_i)}{\sqrt {m_s}}$ and $h'(x)=\sum_{i=1}^{m_s}\frac{\alpha'_iy_ik(x,x_i)}{\sqrt {m_s}}$ to be the noiseless/noisy classifier,
		\begin{align*}
			|h'(x)-h(x)|\leqslant\epsilon
		\end{align*}
		the measurement shots for kernel matrix $K$ scaling as $O(M^4/\epsilon^{2})$.
	\end{lemma}
	\begin{proof}
		To guarantee $\|K'-K\|_F\leqslant\varepsilon'$,   each matrix entry should take $O(M^2\varepsilon'^{-2})$ shots \cite{qfm1}. For each new datum $x$,
		\begin{align*}
			|h'(x)-h(x)|&=|\sum_{i=1}^{m_s}\frac{\alpha_iy_ik(x,x_i)}{\sqrt {m_s}}-\sum_{i=1}^{m_s}\frac{\alpha'_iy_ik(x,x_i)}{\sqrt {m_s}}|\\
			&\leqslant\sum_{i=1}^{m_s}|\frac{\delta_ik(x,x_i)}{\sqrt {m_s}}|\\
			&\leqslant\|\delta\|_2,
        \end{align*}
with $\delta_i=\alpha'_i-\alpha_i$. The upper bound of $\|\delta\|_2$ depends on the minimum eigenvalue of $K$. Due to the positive definiteness of kernel matrix $K$, we have $\lambda_{min}(K+\frac{1}{\gamma}I)\geqslant\frac{1}{\gamma}$ is lower bounded by a constant \cite{nisq}. Then Lemma \ref{lemma1} gives
        \begin{align*}
			\|\delta\|_2\leqslant O(\frac{M}{\sqrt R})\leqslant\epsilon.
		\end{align*}
		 Due to the symmetry of the kernel matrix and the trivial diagonals, $\frac{M^2-M}{2}$ matrix entries have to be estimated. Therefore, a total of $O(M^4/{\epsilon^2})$ shots is sufficient for restricting $|h'(x)-h(x)|\leqslant\epsilon$.
	\end{proof}
	
	Assuming trial datum $x$ is taken from a uniform distribution and abundant, and $p(x)$  be the probability density function of $h(x)$ that satisfying
	\begin{align*}
		\int_{-1}^1p(x)dx=1,
	\end{align*}
	and $p'(x)$ corresponds to $h'(x)$. Due to statistical errors during the quantum kernel estimation phase, the classifier $h'(x)$ we use has a maximum error $\epsilon$ compared to the ideal classifier $h(x)$ which leads to a maximum misclassification rate $\int_0^{\epsilon}p(x)dx$ (integral between two black dashed lines of FIG.\ref{fig3a}).
	
	The aforementioned analysis is under the condition that possessing the exact value of $h'(x)$, however, we merely have access to an estimated approximate value $\widetilde h'(x)$. 
	In general, the accuracy requirement of the training stage is higher than classification, as slightly incorrect classifiers can lead to unpleasant classification accuracy. We provide a theoretical analysis below.

		In the case $\varepsilon\geqslant\epsilon$ (FIG.\ref{fig3a}),
		\begin{align*}
			1-\int_{-\varepsilon}^{\varepsilon}p(x)dx
		\end{align*}
		is the probability of correct classification, and $\int_{-\varepsilon}^{\varepsilon}p(x)dx$ is inconclusive. However, with the increase of measurement shots, i.e. $\varepsilon<\epsilon$ (FIG.\ref{fig3b}), this will give the correct classification rate
		\begin{align*}
			1-\int_{-\epsilon}^{\varepsilon}p(x)dx,
		\end{align*}
		with $\int_{-\varepsilon}^{\varepsilon}p(x)dx$ inconclusive and $\int_{-\epsilon}^{-\varepsilon}p(x)dx$ wrong. The above analysis concludes that the benefits of increasing measurement shots during the classification stage are relatively minimal since the accuracy of classification is always affected by the errors $\epsilon$ stemming from quantum kernel estimation.

	As the scale of data increases, no matter what kind of quantum classifiers, the probability density function becomes increasingly sharp as long as the query state is evenly weighted (blue line in FIG.\ref{fig3c}). Some of these QSVMs can circumvent this dilemma by constructing the partially evenly weighted trail state---only for support vectors. Since the inconclusive rate $\int_{-\varepsilon}^{\varepsilon}p(x)dx<\int_{-\varepsilon}^{\varepsilon}p_{LS}(x)dx$, building up query state according to support vectors is necessary.
	
\section{Quantum iterative multi-classifier}\label{section5}
    In this section, we analyze the essence of multi-classification algorithms and propose a quantum iterative multi-classifier framework.
    \subsection{Principle analysis and preparation}
    Classification problems are of great significance, especially multi-classification. The vast majority of multi-classifiers are derived from the generalization of binary classification. One can use one-versus-one or one-versus-rest strategies to transform binary classifiers into multi-classification.

    When applying one-versus-one to $L$ classification, one needs to prepare $\frac{L(L-1)}{2}$ binary classifiers for each pair of different categories. There is a voting session after each classifier, and the eventual label is the one that earns the most votes. Adapting to the ensemble model, the above process can be simplified as finding the biggest trace distance between the $L$ ensembles and the trial density matrix, that is
   \begin{align*}
\arg\max \limits_{j\in L}\{{\rm Tr}(\rho^j\rho)|\rho^j=\sum_ip_i^j|\psi(x_i^j,\theta^*)\rangle\langle\psi(x_i^j,\theta^*)|,\\ \rho=|\psi(x,\theta^*)\rangle\langle\psi(x,\theta^*)|\}.
   \end{align*}

    When adopting a one-versus-rest approach, $L$ classifiers are sufficient, and each one separates $y_j$ from the others. The final label is given by the unique classifier that chooses $y_j$ as positive. The above behavior can be explained as finding out $j\in L$ that satisfies
    \begin{align*}
    (\rho^j\rho)^{ovr}={\rm Tr} (\rho^j\rho)-\frac{1}{L-1}\sum_{k=1,k\ne j}^L{\rm Tr}(\rho^k\rho)>0.
    \end{align*}
The one-versus-rest model can also be adapted to SVM directly. More generally, no matter one-versus-one or one-versus-rest is chosen, the classification result can be interpreted as finding the largest element in sets $\{{\rm Tr}(\rho^j\rho)\}_{j=1}^L$ and $\{(\rho^j\rho)^{ovr}\}_{j=1}^L$.
    \subsection{Algorithm Flow}
    A trivial option is to calculate each element in turn, then the answer will be revealed naturally. However, this strategy is costly because calculating each element requires a large amount of measurement. Another artful way to read out the classification result is amplitude amplification as Grover search algorithm owns the capability to flip amplitude to the target qubit, then the qubit can be read out with a high probability---usually greater than 50\%. To complete the above task, one needs to store all elements in the set proportionally in the amplitude of the quantum state. That is to say, $U_L$, $U_x$  and their inversion need to be prepared
    \begin{align}
    &U_L|0...0\rangle=\sum_{j=1}^L\frac{1}{\sqrt L}|j-1\rangle\sum_{i=1}^{M_j}\sqrt{p_i}|i-1\rangle|\psi(x^j_i,\theta^*)\rangle,\label{e8}\\
    &U_x|j-1\rangle|0...0\rangle=\sum_{j=1}^L|j-1\rangle\sum_{i=1}^{M_j}\frac{|i-1\rangle|\psi(x,\theta^*)\rangle}{\sqrt{M_j}},\label{e9}\\
    &U_x^{\dagger}U_L|0...0\rangle=|\phi_0\rangle=|\phi_1,...,\phi_2,...,...,\phi_L,...\rangle,\notag
    \end{align}
where the classification result $\{{\rm Tr}(\rho^j\rho)\}_{j=1}^L$ is stored in $\{\phi_j\}_{j=1}^L$ proportionately. For the convenience of explanation, we introduce a swap matrix $S$ such that $S|\phi_0\rangle=|\phi\rangle=|\phi_1,\phi_2,...,\phi_L,{\rm junk}\rangle$.

    The workflow of a single Grover's search algorithm is as follows:(i)implement a phase inversion operator $T={\rm diag}(-1_1,-1_2,...,-1_L,1_{L+1},1,...,1)$ to invert the first $L$ phases of $|\phi\rangle$; (ii) apply $U_G^{\dagger}$ to $T|\phi\rangle$, $U_G=SU_x^{\dagger}U_L$; (iii) reverse the state by $|0...0\rangle$, that is, implement an operator $O={\rm diag}(1,-1,-1,...,-1)$; (iv) apply $U_G$ to step(iii). It should be noted that here $U_G|0...0\rangle=|\phi\rangle$. Repeat the above operations for $R\in[1,\frac{\sqrt{5M_j}\pi-1}{2}]$ times roughly, the phases of classification results can be amplified. The proof of the number of iterations is provided in Appendix \ref{apd}.

	\section{Numerical  verification}\label{section6}
\begin{figure*}[t]
		\centering
        \subfigure[]{\includegraphics[width=0.4\textwidth]{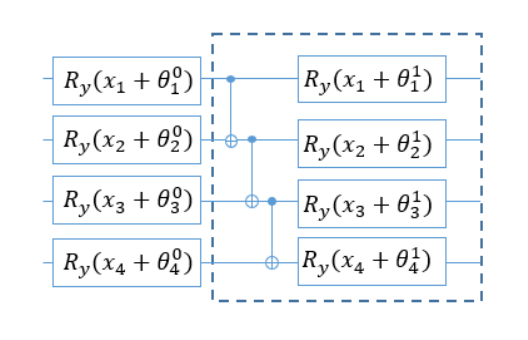}\label{fig4a}}\subfigure[]{\includegraphics[width=0.5\textwidth]{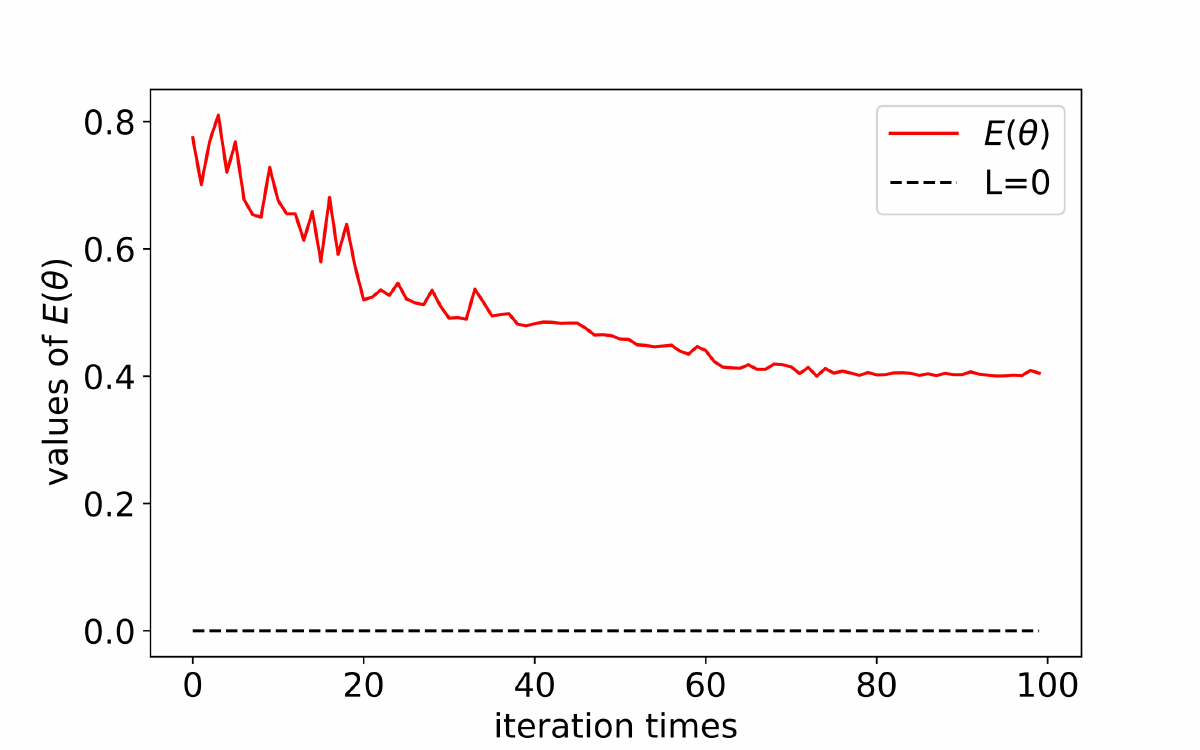}\label{fig4b}}
		\subfigure[]{\includegraphics[width=0.4\textwidth]{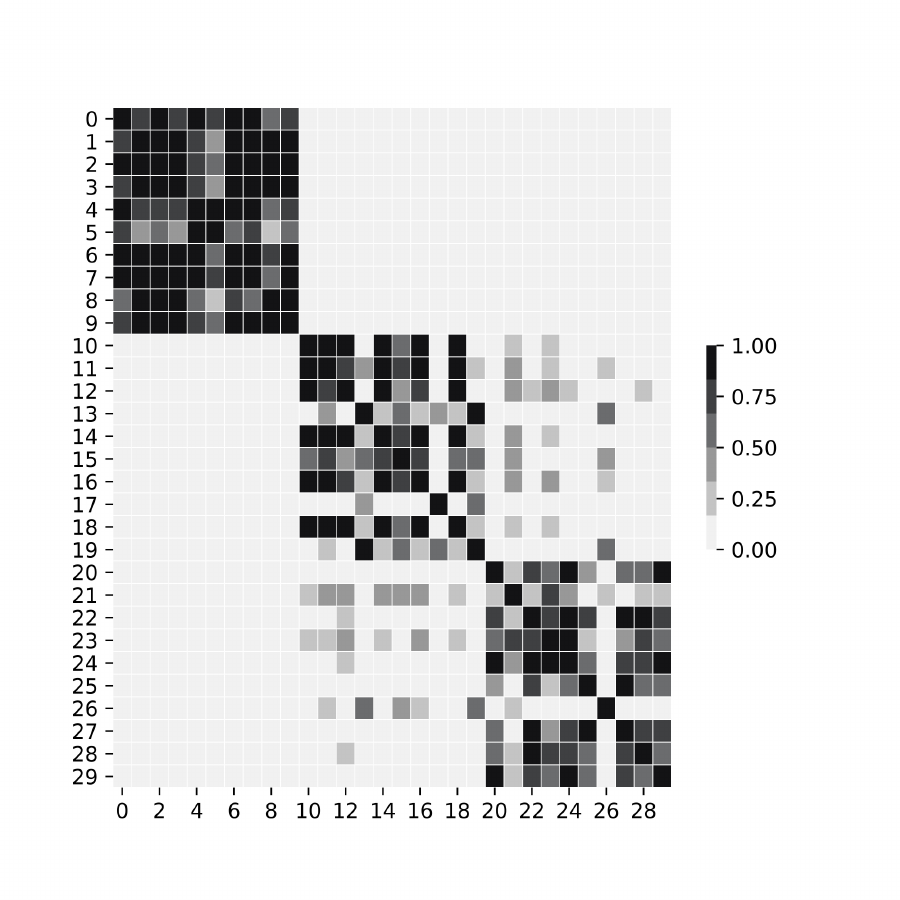}\label{fig4c}}\subfigure[]{\includegraphics[width=0.4\textwidth]{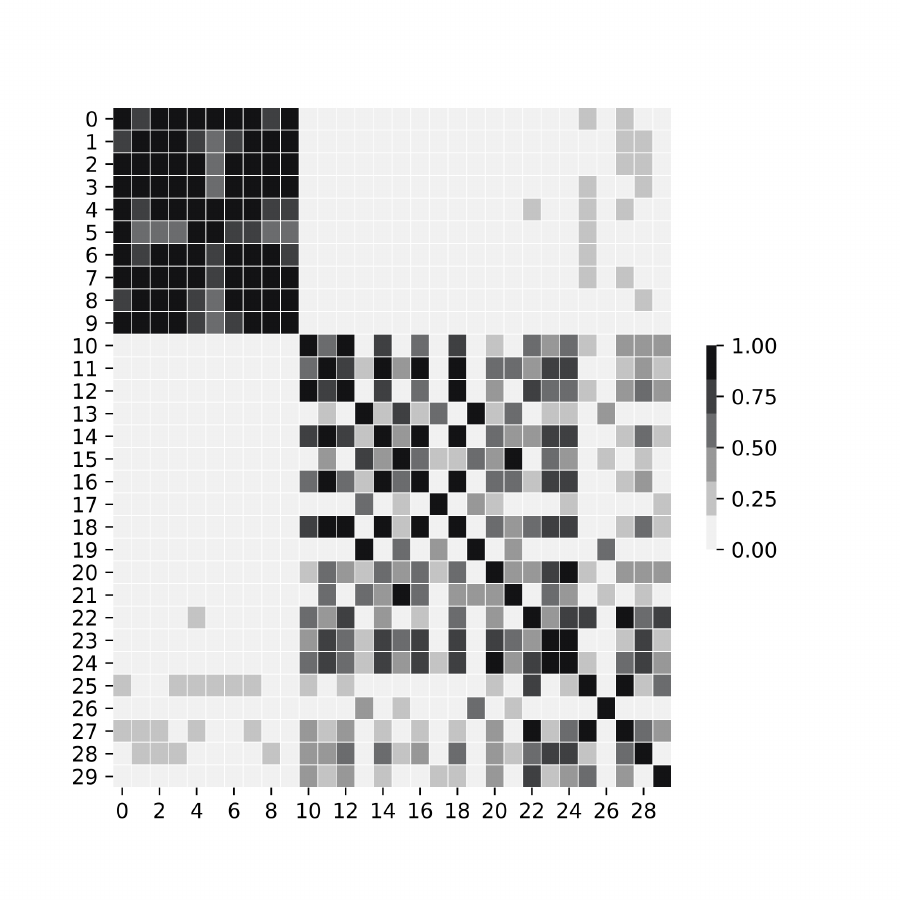}\label{fig4d}}
		\caption{ Numerical {verification} of trainable quantum feature mapping. (a)A basic TQFM framework used in the simulations. (b) The iteration value of Eq.(\ref{e9}), and the starting point is $\theta =0$. (c,d) Visualization of overlaps between training samples. (c) is the kernel matrix of training data after training, i.e. $\theta=\theta^*$, and (d) is before training, i.e. $\theta=0$. }
	\end{figure*}
In this section, we carry out the corresponding numerical simulations for the proposed algorithms to validate the feasibility in \textit{Python qiskit} package. We target the IRIS data set with $L=3$ labels and 4 features, and each class has 50 samples.
    In the classical optimization stage, taking the presence of measurement uncertainty into consideration, we adopt \textit{COBYLA} optimizer for convenience. Quantum gradient optimization methods can also be a good choice \cite{qgd1,qgd2} as one could evaluate the derivatives of quantum expectation values from the output of variational quantum circuits.
    \subsection{Trainable quantum feature mapping}

To verify the feasibility and effectiveness of trainable quantum feature mapping, we sample 10 data points from each class as training data. Each feature plus an adjustable parameter is compiled into the amplitude of a parameterized quantum circuit. The layout of this simulation is shown in FIG.\ref{fig4a}, the dashed box(also known as a layer) comprises entanglement gates and the re-uploading gates. A better clustering effect can be achieved through training the circuit. In this simulation,the loss function becomes
    \begin{align}
    E(\theta)=1-\frac{1}{3}\sum_{j=1}^3\frac{1}{10}\sum^{10}_{i=1}|\langle\psi(x^j_i,\theta)|y_j\rangle|^2,\label{e9}
    \end{align}
where $|y_j\rangle=|j-1\rangle|\cdot\rangle$. 	

    The iterative process of optimizing function values is shown in FIG.\ref{fig4b}. Only 2-layers of the repetition parts and the simple $R_y$ rotations are used, leading the value of the loss function to only reach 0.4. We can achieve better training effects by increasing the complexity of the circuit, such as increasing the depth of the circuit and using $R_z$-$R_y$-$R_z$ encoding. We will not delve too deeply into this matter.

The iteration of the loss function is completed with the optimal parameters $\theta^*$. To present the training effects more directly, we visualized the overlaps between training data by drawing the hotspot matrices(FIG.\ref{fig4c},\ref{fig4d}). The training process further orthogonalized Class 1 with Class 2, 3. What is more obvious is that Class 2 and Class 3, which were originally difficult to distinguish, have significantly improved their distinguishability after training. It can be seen that using the trained kernel matrix for machine learning algorithms, such as SVM, will definitely outperform the untrained one in terms of performance.
    \subsection{Explicit application of training quantum feature mapping}
     For simplicity, we perform a binary classification validation and choose \textit{explicit approach} as a comparison. We choose Class 2 and Class 3 for the {simultion} because they are linearly inseparable, which better tests the performance of the classifier. We adopt the feature map $U_{\Phi(x)}=U_{\phi(x)}H^{\otimes n}U_{\phi(x)}H^{\otimes n}$ defined in \cite{qfm1} to generate state and
    \begin{align*}
    U_{\phi}(x) &= \bigotimes_{i=1}^nR_z(x^{(i)}).
    \end{align*}
Both the layouts of $W(\theta)$ and TQFM employ the $R_z$-$R_y$-$R_z$ rotations, and $R_y$ is taken for TQFM to be an object of reference. We select 10 samples for training, then classify all 100 data and conduct {simultions} with different layers. The starting point is random, thus we conduct 5 {simultions} to take the average value to eliminate the interference of randomness.

    From the {simultion} results as shown in Table \ref{t1}, we deduce that TQFM has stronger classification performance than \textit{explicit approach} in \cite{qfm1}, of course, we cannot rule out the dependence on data. But this at least reveals the advantages of TQFM in certain aspects, that is, TQFM can choose to use multiple layers to prepare the kernel matrix for classification, in order to further improve classification performance. Meanwhile, TQFM can also achieve a considerable success rate(over 80\%) without introducing entanglement gates.
\begin{table}[b]
\caption{A comparison of explicit applications. SR and $E$ mean success rate and loss function, respectively.  The lowercase letters y and z represent Pauli rotation used in the validation, and $r$ represents the number of layers. }\label{t1}
\begin{tabular}{lllllllll}
\toprule 
&\multicolumn{2}{c}{$r=0$}&&\multicolumn{2}{c}{$r=1$}&&\multicolumn{2}{c}{$r=2$}\\
\makecell[c]&SR(\%)&$E$&&SR(\%)&$E$&&SR(\%)&$E$\\
\midrule
\makecell[l]{TQFM(y)}&68&0.41&&81.6&0.31&&86&0.24\\
\makecell[l]{TQFM(zyz)}&83.2&0.34&&86&0.25&&91.8&0.21\\
\makecell[l]{QFM(zyz)}&63.6&0.45&&84.4&0.34&&89.8&0.25\\
\bottomrule
\end{tabular}
\end{table}

As a promotion, we prepare an ensemble model to test the subsequent use effect of TQFM that employs $R_y$ rotations. We use the optimal parameters $\theta^*$ obtained from Table \ref{t1}, and select 4 samples to prepare the ensemble model, then classify 100 data. The average classification success rates for $r$ = 0, 1, and 2 are 92\%, 91.8\%, and 94.8\%, respectively. It can achieve a classification accuracy of up to 95\% when using only 1/5 of the training samples with the simple use of $R_y$ rotation.

    \subsection{Variational quantum support vector machine}
	\begin{figure}[t]
		\centering
        {\includegraphics[width=0.5\textwidth]{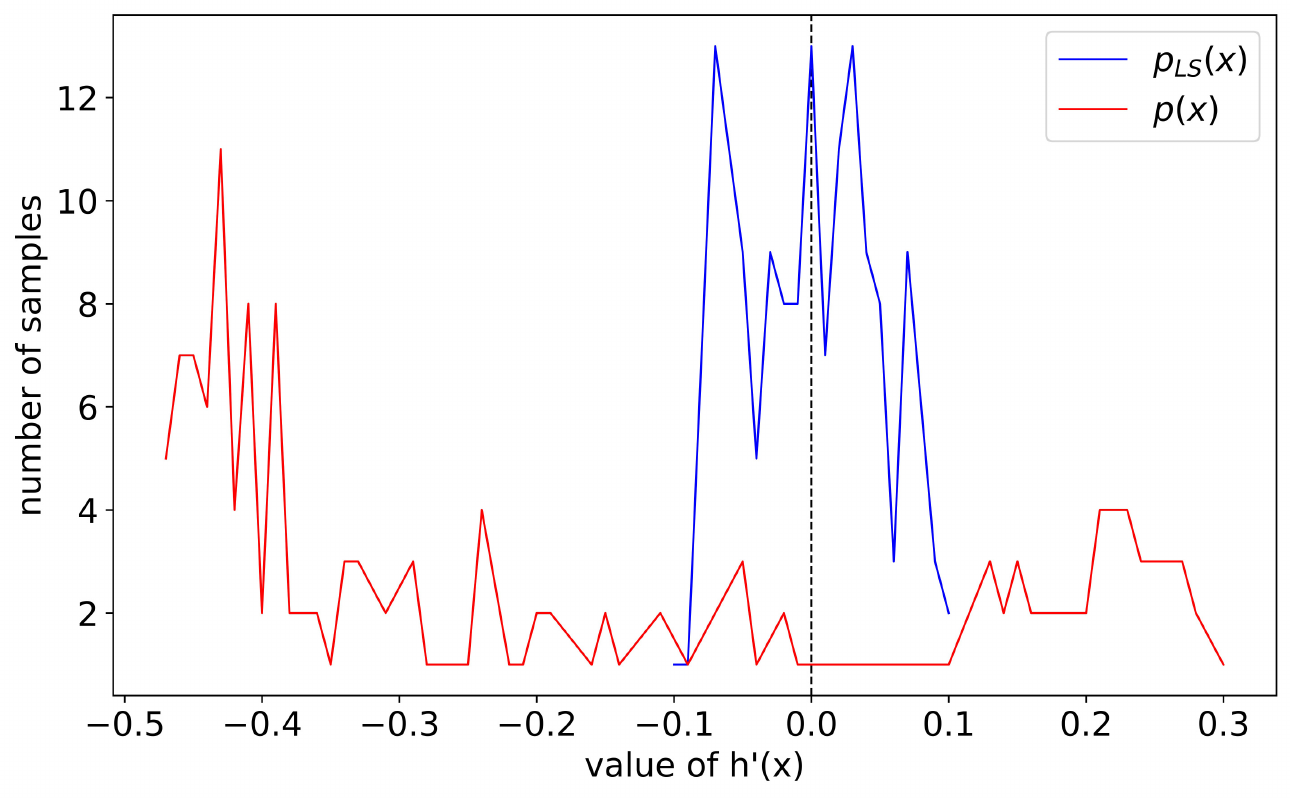}}
		\caption{Sample density distribution where the function value is rounded to two decimal places. Values greater than 0 are positive, while values less than 0 are negative.}\label{fig5}
	\end{figure}
     For comparison, we conduct two numerical simulations to demonstrate the distinguishability of SV-QSVM and choose LS-QSVM \cite{QSVM} as a representative of quantum classifiers that employ evenly weighted superposition during classification. 	At present, HHL algorithm can only be executed for several qubits, so we choose a variational quantum algorithm as an alternative \cite{mine}.

    Because the focus is on comparing branching, we only select 8 samples (4 for Class 1, 2 for Class 2, 2 for Class 3, Class 1 is the positive class, while the other two classes are negative) for training, and set the bias to be 0. When classifying, we test all 150 data points, and their classification distribution diagram is shown in FIG.\ref{fig5}.   To approximate the value of $h(x)$ as accurately as possible, we set 4095 measurements. It is obvious that SV-QSVM has much greater distinguishability than LS-QSVM, which not only improves distinguishability but also reduces error rates, with actual accuracy rates of 100\% and about 92\%, respectively. Exporting the results for analysis, we find that the reason why the accuracy of the latter always cannot reach 100\% is that the values located in the interval $[-0.05,0.05]$ where the sign function is prone to errors account for a high proportion.
     \subsection{Quantum iterative multi-classifiers}

In the numerical validation of multi-classifier, we conduct the one-versus-one strategy as an example. We take 12 samples (4  per class) for classification and omit the training stage because assigning equal weights to these samples can also be used as a trivial training result. The Eq.(\ref{e8}) of this simulation is
    \begin{align*}
    U_L|0...0\rangle=\sum_{j=1}^3\frac{1}{\sqrt 3}|j-1\rangle\sum_{i=1}^4\frac{1}{2}|i-1\rangle|\psi(x_i^j,\theta^*)\rangle.
    \end{align*}

	\begin{figure}[t]
		\centering
        {\includegraphics[width=0.5\textwidth]{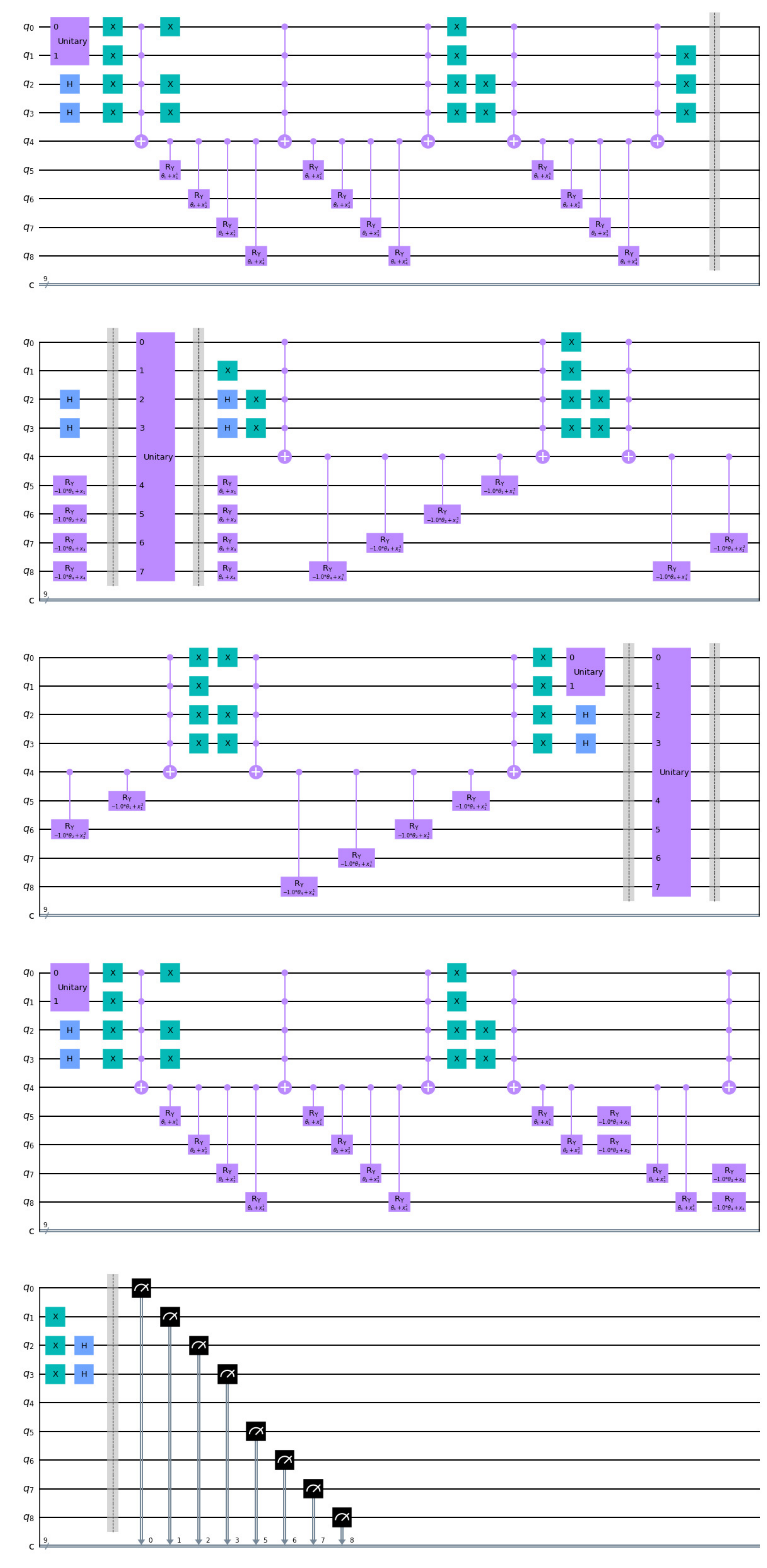}}
		\caption{Circuit diagram of quantum  iterative multi-classifier with 1 iteration. $^0_1\rm{Unitary}|00\rangle=(\frac{\sqrt 3}{3},\frac{\sqrt 3}{3},\frac{\sqrt 3}{3},0)$. The functions of the two $^{0,1,2,3}_{4,5,6,7}\rm{Unitary}$ are described by the third and fifth blocks. }\label{fig6}
	\end{figure}	
The schematic diagram of the {simulation} is shown in FIG.\ref{fig6}. For the sake of simplicity, we only reserve 3 samples. There are a total of 6 barriers that divide the circuit into 7 blocks.  $U_L$ is the first block, and $U_x$, which stores trial data, is the second block. The function of the third block is to mark the target phase, which is to flip the amplitude. The fourth block is the inverse of $U_xU_L$. The fifth block is a flipping operator, that is flipping all phases around $|0...0\rangle$. The sixth block is a repetition of blocks 1 and 2, and the last block is for measurement. There are a total of 9 qubits, with $q_0$ and $q_1$ serving as labels, $q_2$ and $q_3$ storing the amplitude of sample data, $q_4$ being a working qubit, and $q_5$-$q_8$ being the quantum feature mapping.

\begin{table}[t]
\centering
\caption{The probabilities and sample distribution of reading out IRIS data set with 10000 shots for each sample. r is the layers used in the TQFM stage. The column of $<$0.5 excludes the erroneous data.}\label{t2}
\begin{tabular}{llllllllll}
\toprule 
Classes&\multicolumn{3}{c}{Before iteration}&\multicolumn{3}{c}{After iteration}&\multicolumn{3}{c}{Sample distribution}\\
\makecell[l]&max&min&ave&max&min&ave&$>$0.5&$<$0.5&error\\
\midrule
\makecell[l]1(r=0)&0.33&0.18&0.30&0.96&0.85&0.91&50&0&0\\
\makecell[l]2(r=0)&0.29&0.09&0.23&0.65&0.23&0.47&20&29&1\\
\makecell[l]3(r=0)&0.30&0.10&0.24&0.73&0.22&0.39&11&33&6\\
\makecell[l]1(r=1)&0.33&0.12&0.28&0.96&0.66&0.81&50&0&0\\
\makecell[l]2(r=1)&0.30&0.09&0.23&0.95&0.39&0747&44&6&0\\
\makecell[l]3(r=1)&0.28&0.07&0.21&0.73&0.22&0.39&31&9&10\\
\makecell[l]1(r=2)&0.32&0.12&0.26&1.00&0.71&0.94&50&0&0\\
\makecell[l]2(r=2)&0.26&0.01&0.18&0.97&0.08&0.76&44&5&1\\
\makecell[l]3(r=2)&0.26&0.02&0.19&1.00&0.14&0.73&42&1&7\\
\bottomrule
\end{tabular}
\end{table}

    In the classification validation stage, we classified all 150 data with a success rate of 95.3\%, 93.3\%, and 94.7\%, respectively. The detailed distribution results are shown in TABLE \ref{t2}. We speculate that there may be two possible reasons why the classification success rate did not reach 100\%: (i) There is too little {simulated samples} we used, and only 4 samples in each class. Nonetheless, it exceeds the success rate of 92.5\% in Ref.\cite{qmc} (ii) The trainable quantum feature mapping structure we used is too simple and does not cluster perfectly during the training phase. From Table \ref{t2}, we can see that the amplitude can be significantly improved through iteration. The proportion of probability exceeding 0.5(corresponding to an amplitude greater than 0.71) has increased from 81 to 136. This extremely reduces the burden of reading out results.

	\section{Conclusion and Prospect}\label{section7}
	In this paper, we propose a trainable quantum feature mapping strategy that owns the power to deal with the non-linear problems through data re-uploading. A well-trained TQFM can be directly used for classification or as a subroutine in machine learning algorithms \cite{TQFMout,TQFMout1}. We have conducted in-depth research on quantum classifiers equipped with TQFM and proposed the support vectors based variational quantum support vector machine algorithm that shows a better distinguishability in the classification stage.  We rigorously analyze the errors and advantages of SV-QSVM. As a promotion,  we demonstrate how to use a quantum search algorithm to iteratively read out multi-classification results with high probability.

    We have numerically validated the proposed algorithms in \textit{qiskit} package, all the results are consistent with our analysis. We have trained the IRIS dataset with the simplest TQFM layout, and the resulting kernel matrix exhibits considerable clustering performance. We verify the preferable distinguishability by introducing LS-QSVM as a comparison. The experimental results not only demonstrate the better distinguishability of SV-QSVM, but also surpass LS-QSVM in terms of classification success rate. In respect of multi-classification, through iteration, the proportion of stored classification results with amplitude greater than 0.71 is greater than 90\%. This is enough to ensure that one can accurately read out the results with a lower number of measurements.

     It is necessary to mention the barren plateau phenomenon during the optimization process, as it is an open question that variational quantum algorithms need to face. As the number of qubits increases, the gradient becomes concentrated, causing the training to suffer from the curse of the barren plateau  \cite{barren}. If the loss functions can be reduced to local loss functions, then their gradients vanish  polynomially rather than exponentially in the number of qubits \cite{barren1}. The recent advancements on the issue of barren plateaus have deeply explored this type of problem \cite{barren2}.

     The optimization stage of SV-QSVM may be undesirable for nowadays NISQ devices. However, quantum annealers manufactured by D-Wave Systems are available with about 2000 qubits \cite{dwave,dwave1}. They automatically produce a variety of close-to-optimal solutions to an optimization problem.  That is, the optimization  process of SVM can be accomplished on d-wave annealers unhurriedly \cite{dwave2}. This has broadened the path for current quantum machine learning.

	\section*{Acknowledgments}
	We thank Jin-min Liang for the helpful discussions and suggestions. This work is supported by the Shandong Provincial Natural Science Foundation for Quantum Science No. ZR2021LLZ002 and the Fundamental Research Funds for the Central Universities No. 22CX03005A.
	
	\begin{appendix}
		\section{Support vector machine\label{apa}}
		Here, we briefly review the classical support vector machine for binary classification. The mission of an SVM is assigning a label to test datum $x\in S\in\mathbb{R}^n$ depending on the decision function stems from training samples $\{x_i,y_i\}_{i=1,2,...,M}\in T\times\{+1,-1\}\in\mathbb{R}^n\times{R}$.  An optimal decision function can be generated by maximizing the distance $\frac{2}{\|w\|}$ between hyperplanes with opposite labels, equal to
		\begin{align}
			&\min \frac{1}{2}\|w\|^2 \label{8}\\
			\notag {\rm s.t.}&\  y(w^Tx_i+b)\geqslant1,\\
			\notag &\  i=1,2,...,M.
		\end{align}
		Returns $w$ for classifying through $y=\sign{w^Tx+b}$. The above assumption is under ``ideal conditions", i.e. $T$ is linearly separable. However, the most practical dataset is non-linear separable, thus Eq.(\ref{8}) will get stuck with ``bad conditions". By introducing slack variables, the above predicament can be overcome
		\begin{align}
			\min &\frac{1}{2}\|w\|^2+\frac{\gamma}{2}\sum_{i=1}^M\xi_i^2 \label{9}\\
			\notag {\rm s.t.}\  &y(w^Tx_i+b)\geqslant1-\xi_i,\\
			\notag &\  i=1,2,...,M.
		\end{align}
		Although the above optimization can be directly solved, there is a more efficient choice by introducing the Lagrange multiplier, then deriving Lagrange dual form
		\begin{align}
			\max_{\alpha} &\sum_{i=1}^M\alpha_i-\frac{1}{2}\sum_{i,j=1}^M\alpha_i\alpha_jy_iy_jK_0(x_i,x_j)-\frac{1}{2\gamma}\sum_{i=1}^M\alpha_i^2\notag\\
			&{\rm s.t.}\ \alpha_i\geqslant0,\ \sum_{i=1}^M\alpha_iy_i=0,\ i=1,2,...,M.\notag
		\end{align}
		Gives the optimization result
		\begin{align*}
			y = \sign{\sum_{i=1}^M\alpha_iy_iK_0(x,x_i)+b}.
		\end{align*}
		By adding $l_2$-regularization term of bias $\frac{\lambda}{2}b^2$ to Eq.(\ref{9}) can acquire a relaxed decision function
		\begin{align}
			y=\sign{\sum_{i=1}^M\alpha_iy_i[K_0(x,x_i)+\frac{1}{\lambda}]}\label{9.1}.
		\end{align}
		with $b=\frac{\sum_{i=1}^M\alpha_iy_i}{\lambda}$.
		
		The primal Lagrangian is given by
		\begin{align*}
			L_P=\frac{1}{2}\|w\|^2+\frac{\gamma}{2}\sum_{i}\xi_i^2+\frac{\lambda}{2}b^2-\sum_i\alpha_i(yw^Tx+yb-1-\xi_i),
		\end{align*}
		the dual optimal solutions Eq.(\ref{9.1}) can be connected via the Karush-Kuhn-Tucker (KKT) conditions
		\begin{align*}
			w&=\sum_i\alpha_ix_iy_i,\\
			0&=\sum_i\alpha_iy_i,\\
			\xi_i&=\frac{\alpha_i}{\gamma},\\
			b&=\sum_i\frac{\alpha_iy_i}{\lambda}.
		\end{align*}
      \section{Random sampling}\label{apb}
The workflow of random sampling begins with refining the $M\times{M}$ kernel matrix $K$ in terms of non-zero elements in the form
	\begin{align*}
		K &= \sum_{(i)}K_{x^{(i)},y^{(i)}}|x^{(i)}\rangle\langle y^{(i)}|\\
		&= \sum_{(i)}K_{x^{(i)},y^{(i)}}|x^{(i)}_1\rangle\langle y^{(i)}_1|\otimes|x^{(i)}_2\rangle\langle y^{(i)}_2|\otimes...\otimes|x^{(i)}_M\rangle\langle y^{(i)}_M|
	\end{align*}
	where $(i)$ represents the index of entries and $K_{x^{(i)},y^{(i)}}$ is the matrix entry at $(x^{(i)},y^{(i)})$. In addition, we have
	\begin{align*}
		|0\rangle\langle0|&=(I+Z)/2,\\
		|0\rangle\langle1|&=(X+iY)/2,\\
		|1\rangle\langle0|&=(X-iY)/2,\\
		|1\rangle\langle1|&=(I-Z)/2,
	\end{align*}
	and $K$ could be expanded as Pauli operators
	\begin{align*}
		K = \sum_{(i)}K_{x^{(i)},y^{(i)}}\sum_{j^{(i)}}\frac{1}{2^M}\sigma_{j^{(i)}},
	\end{align*}
	with $\sigma_{j^{(i)}}$ being one of $2^M$ Pauli strings. At the end, to randomly measure $K$, we first sample $(i)$ from $|K_{x^{(i)},y^{(i)}}|/q(i)=\sum_{(i)}|K_{x^{(i)},y^{(i)}}|$ and then toss $m$ unbiased coins to decide $j^{(i)}$ for each $\sigma_{j^{(i)}}$. In particular, according to Hoeffding's inequality, we need $O(\log(\delta^{-1})\Lambda^2/\varepsilon^{2})$ samples to an error $\epsilon$ with failure probability $\delta$, where $\Lambda=\sum_{j}\lambda_j$.

		\section{Quantum least squares support vector machine\label{apc}}
		The core idea is the least square method, which transforms the quadratic process into solving linear equations. The introduction of slack variables replaces the inequality constraint of Eq.(\ref{9}) with the equality constraint, the restriction of $\xi_i\geqslant0$ has been removed, this time $\xi_i$ represents the degree to which the sample does not satisfy the constraint. After a series of algebraic operations gives
		\begin{gather}
			F\begin{pmatrix}
				b\\\vec\alpha
			\end{pmatrix}
			=\begin{pmatrix}
				0&\vec1^{\rm{T}}\\\vec1&K_0+\gamma^{-1}I
			\end{pmatrix}
			\begin{pmatrix}
				b\\\vec\alpha
			\end{pmatrix}
			=\begin{pmatrix}
				0\\\vec{y}
			\end{pmatrix}. \label{10}
		\end{gather}
		Here, the matrix $F$ is $(M+1)\times(M+1)$ dimensional, $K_0(x_i,x_j)=\langle\psi(x_i)|\psi(x_j)\rangle$, $\vec{y}=(y_1;y_2;...;y_M)$, and $\vec1=(1;1;...;1)$. The hyperplane parameters can be obtained by solving Eq.(\ref{10}) in a quantum manner such as the HHL algorithm, variational quantum-classical algorithm et.
		
		The label for datum $x$ will be revealed by
		\begin{align*}
			y&=\sign{\langle{v}|u\rangle},\\
			|u\rangle&=\frac{1}{\sqrt{N_u}}(b|0\rangle|0\rangle+\sum_{i=1}^M\alpha_i|i\rangle|\psi(x_i)\rangle),\\
			|v\rangle&=\frac{1}{\sqrt{N_{x}}}(|0\rangle|0\rangle+\sum_{i=1}^M|i\rangle|\psi(x)\rangle),
		\end{align*}
		with $N_u$ and $N_{x}$ normalization factor.
\section{Estimation of iteration times} \label{apd}
    In the estimation of iteration times, we proceed under the assumption of \textbf{Theorem} \ref{th1}, and set each $M_j$ which is the number of category $j$ is equal.  The final state before applying amplitude amplification can be refined as
    \begin{align*}
    |\phi_0\rangle=a|a\rangle+\sqrt{1-a^2}|b\rangle,
    \end{align*}
where $|a\rangle$ is the target phase, that is, the amplitude of these phases needs to be amplified. Totaling $L$ phases, and each one stores the result. The sum of $L$ amplitudes, i.e. the value of $a$, determines the number of iterations, is roughly $g(x)=\sum_{i=1}^{M_j}\frac{\alpha_ik(x_i,x)}{\sqrt M_j}$. And $|b\rangle$ is a summary of the phase of storing junk information. The upper limit of $a$  is easily determined by
    \begin{align*}
    a&=\sum_{i=1}^{M_j}\frac{\alpha_ik(x_i,x)}{\sqrt M_j}\\
     &\leqslant\|\alpha_i\|\|k(x_i,x)\|\\
            &=\sqrt{\frac{k^2(x_1,x)+k^2(x_2,x)+...+k^2(x_{M_j},x)}{M_j}}\\
            &=\frac{1}{\sqrt 3},
    \end{align*}
where the second line uses Cauchy-Schwarz inequality and the third line uses the \textit{Central Limit Theorem}. Assuming $a$ is a positive number, and the sign of $a$ does not affect the number of iterations.

To estimate the lower bound of $a$, we need to use $\alpha_i\propto\frac{1}{\sqrt{M_j}}$, then $g(x)$ flows a normal distribution $N(0,\frac{1}{3M_j})$. So we have
    \begin{align*}
    &\int_0^ag(x)dx=\frac{1}{4}\\
    &\int_0^a\frac{1}{\sqrt{2\pi}\sigma}e^{-\frac{x^2}{2\sigma^2}}dx=\frac{1}{4}\\
    &a=\sqrt{\frac{2\ln{\frac{4}{3}}}{3M_j}}\approx \frac{1}{\sqrt{5M_j}}.
    \end{align*}
The reason we chose the interval between the median and maximum values is that the function is highly dense near the median point, and a well-trained quantum feature mapping has good distinguishability.

Then the iteration equation
    \begin{align*}
    \frac{a}{2}+Ra=\frac{\pi}{2}
    \end{align*}
 gives $R\in[1,\frac{\sqrt{5M_j}\pi-1}{2}]$.\\

	\end{appendix}

\end{document}